\documentclass[letterpaper, 10 pt, conference]{ieeeconf}
\IEEEoverridecommandlockouts
\usepackage{cite}
\usepackage{amsmath} 
\usepackage{amssymb}  
\usepackage{amsfonts}  
\usepackage{graphics}
\usepackage{pdflscape}
\usepackage{rotating}
\usepackage{epstopdf}
\usepackage{subcaption}
\usepackage{txfonts}

\newtheorem{prop}{Proposition}

\def\L2e{{\cal L}_{2e}}

\def\begequarr{\begin{eqnarray}}
\def\endequarr{\end{eqnarray}}
\def\begequarrs{\begin{eqnarray*}}
	\def\endequarrs{\end{eqnarray*}}
\def\begarr{\begin{array}}
	\def\endarr{\end{array}}
\def\begequ{\begin{equation}}
\def\endequ{\end{equation}}

\def\begdes{\begin{description}}
	\def\enddes{\end{description}}
\def\begenu{\begin{enumerate}}
	\def\begite{\begin{itemize}}
		\def\endite{\end{itemize}}
	\def\endenu{\end{enumerate}}
\def\lef[{\left[\begin{array}}
	\def\rig]{\end{array}\right]}

\def\begcen{\begin{center}}
	\def\endcen{\end{center}}

\def\ss{\sum_{i=1}^n}

\newtheorem{remark}{Remark}

\usepackage{color}

\usepackage[prependcaption,colorinlistoftodos]{todonotes}

\graphicspath{{}{fig/}}

\title{\LARGE \bf
	Frequency Estimation of Multi-Sinusoidal Signals in Finite-Time*
}

\author{Anastasiia Vediakova$^{1}$, Alexey Vedyakov$^{2}$, Anton Pyrkin$^{2}$, Alexey Bobtsov$^{2}$, Vladislav Gromov$^{2}$ 
	\thanks{*This work was not supported by any organization}
	\thanks{$^{1}$Anastasiia Vediakova is with Department of Computer Applications and Systems, Saint Petersburg State University, Universitetskaya nab., 7/9, 199034, Saint Petersburg, Russia
		{\tt\small vediakova@gmail.com}}%
	\thanks{$^{2}$Alexey Vedyakov, Alexey Bobtsov, Anton Pyrkin and Vladislav Gromov are with the Faculty of Control Systems and Robotics, ITMO University, Kronverksky av., 49, 197101, Saint Petersburg, Russia
		{\tt\small vedyakov@itmo.ru, bobtsov@mail.ru, pyrkin@itmo.ru, 	gromov@corp.ifmo.ru}}%
}

\begin{document}
	
	\maketitle
	\thispagestyle{empty}
	\pagestyle{empty}

	\begin{abstract}
		This paper considers the problem of frequency estimation for a multi-sinusoidal signal consisting of $n$ sinuses in finite-time. The parameterization approach based on applying delay operators to a measurable signal is used. The result is the $n$th order linear regression model with $n$ parameters, which depends on the signals frequencies. We propose to use Dynamic Regressor Extension and Mixing method to replace $n$th order regression model with $n$ first-order regression models. Then the standard gradient descent method is used to estimate separately for each the regression model parameter. On the next step using algebraic equations finite-time frequency estimate is found. The described method does not require measuring or calculating derivatives of the input signal, and uses only the signal measurement. The efficiency of the proposed approach is demonstrated through the set of numerical simulations. 
		
		\textit{Keyword} --- online frequency estimation, continuous-time estimation, finite-time estimator
	\end{abstract}

\section{Introduction}

\label{sec:introduction}

Online frequency estimation for a signal composed of a single or multiple sinusoids is a fundamental theoretical problem. It is widely presented in many practical applications, for example, in vibration control and disturbance rejection systems \cite{Pyrkin2016Compensation}, in precise positioning systems for nanotechnology \cite{nano2008}, in dynamic positioning systems for vessels under external disturbances such as waves, winds and currents \cite{waves2007}. In power systems frequency estimation is used for load balancing, fault detection, enhancing power quality \cite{Xia2012, Phan2016, vedyakov2017globally}. The sensorless speed estimation approach based on online frequency estimation is proposed in \cite{roque2014}.

For known frequencies, estimation of a bias, amplitudes, and phases is a simple linear regression problem. Frequencies appear nonlinearly, and their estimation is much more complicated. This paper focuses on this part. Using estimated frequency values, one can estimate other parameters in a cascaded manner.

The problem is intensively studied in signal processing, instrumentation and measurements, and adaptive control. However, these studies mainly consider different aspects of the problem \cite{Hou2012}. Online estimation adds adaptive properties to the control. Global convergence of the estimates to the true values helps to guarantee the stability of the closed-loop system.

There are relevant results in discrete-time, but the global convergence problem in discrete-time is meaningless due to the critical dependence of the estimate on the sampling time~\cite{Ortega1999}. For example, in \cite{Hasan2013} the authors are assuming that the frequencies are not close to 0 or $\pi$.

In power systems, phase-locked-loop (PLL) algorithms are typically used for frequency estimation, and an extension of the PLL technique is described in \cite{Karimi-Ghartemani2004}, \cite{Karimi-Ghartemani2012}. However, as it is mentioned in \cite{Pin2017}, only local convergence can be proved for these methods.

The advanced results for multi-sinusoidal signals in continuous-time provide global exponential convergence \cite{marino2014}, \cite{pyrkin2015Estimation} and finite-time convergence \cite{chen2016multi}. Another interesting observer-based method is described in~\cite{chen2018observer}, where the number of sinusoids in the measurable signal is supposed to be time-varying.

In this paper, the previously proposed sinusoidal signal frequency estimation method~\cite{med2017sin} is extended to unbiased multi-sinusoidal signals. A solution is based on parameterization with delay operators, which also was used to estimate frequencies of non-stationary sinusoidal signals~\cite{2017ejc}.
After parameterization linear regression model is obtained. Using Dynamic Regressor Extension and Mixing method (DREM)~\cite{Aranovskiy2016}, we split this model into scalar regressions. This step improves transition behavior and allows to obtain estimates at a predefined finite time using a scheme described in~\cite{ortega2019fto}.

The paper is organized as follows. Section~\ref{sec:statement} formulates the problem. Section~\ref{sec:parametrization} presents the parameterization to obtain a linear regression model for a single sinusoid and for a multi-sinusoidal signal. Transformation from the $n$th-order regression to $n$ first-order regression models using DREM is described in Section~\ref{sec:DREM}. Section~\ref{sec:estimator} introduces finite-time estimation schemes for each parameter. In Section~\ref{sec:example} simulation results are presented illustrating the approach efficiency.


\section{Problem statement}
\label{sec:statement}

Consider a measurable multi-sinusoidal signal: 
\begin{align}
	\label{eq:y}
	y(t) =& \ss A_i\sin\left(\omega_i t+\phi_i\right),
\end{align}
where $\omega_i \in \mathbb{R}_+$ are frequencies $\omega_i \neq \omega_j$, $i, j=\overline{1,n}$, $\phi_i$ are phases, $A_i \in \mathbb{R}_+$ are amplitudes, $n$ is the number of harmonics in the signal $y(t)$. The parameters $\omega_i$, $\phi_i$ and $A_i$ are unknown.

The objective is to find the estimates $\hat{\omega}_i(t)$ of the frequencies $\omega_i$ that provide convergence of the errors $\tilde{\omega}_i(t) = \omega_i - \hat{\omega}_i(t)$ to zero at the predefined finite time $t_{ft} > 0$, {\it i.e.}
\begin{align}
\label{eq:objective}
| \tilde{\omega}_i(t) | = 0, \quad
\text{ for } t \geq t_{ft}. 
\end{align}

Our basic assumption is the following.

\noindent {\bf Assumption A1.} 
The lower and upper bounds on the signal frequencies $\omega_i$ are known and equal to $\underline{\omega}$ and $\overline{\omega}$, where~$\underline{\omega}~>~0$.



\section{Parameterization}
\label{sec:parametrization}
In this section we aim to find a linear regression model with measurable variables and a constant vector depending on the signal frequencies.

We introduce an $h$-second delay operator:
\begin{align}
\label{eq:delay_operator}
[Z(\cdot)](t) = \begin{cases}
0, & t < h,\\
(\cdot)(t-h), & t \geq h,
\end{cases}
\end{align}
where $h \in \mathbb{R}_{+}$ is a chosen delay value.

Signals with multiple delays can be represented using this delay operator, as $y(t-kh) = Z^k y(t)$, $k = \overline{1,2n}$.

\begin{remark}
	The delay $h$ is chosen such that
	\begin{align}
	\label{eq:h}
	h < \frac{\pi}{2\overline{\omega}}.
	\end{align}
\end{remark}

Let us find an expression, where the signal $y(t)$ with $n$ harmonics is expressed via $2n$ delayed signals $y(t-h)$, $\ldots$, $y(t-2nh)$. 


\subsection{The case of a single harmonic}
Consider a measurable signal for the case with $n=1$:
\begin{align}
	\label{eq:one_sin}
	y(t) =& A_1\sin\left(\omega_1 t+\phi_1\right).
\end{align}

It is well known~\cite{gromov2017first} that the signal \eqref{eq:one_sin} can be represented as a linear regression model:
\begin{align}
	\label{eq:rel1}
	y(t-2h) + y(t) - 2c_1y(t-h) &= 0,
\end{align}
or with the delay operator \eqref{eq:delay_operator}:
\begin{align}
	\label{eq:z1}
	\left[Z^2 + 1 - 2c_1 Z\right] y(t) &= 0,
\end{align}
where $c_1 := \cos{\omega_1 h}$ and $t \geq 2h$.

Then, the regression model can be written as follows
\begin{align}
\label{eq:regmodel1}
\psi(t) &= \varphi(t) \theta,
\end{align}
where  $\psi(t) \in \mathbb{R}$ is a regressand, $\varphi(t) \in \mathbb{R}$ is a regressor, $\theta \in \mathbb{R}$ is an unknown parameter:
\begin{align}
\label{eq:regmodel1_psi}
\psi(t) &= -y(t-2h) - y(t), \\
\label{eq:regmodel1_phi}	
\varphi(t) &= 2y(t-h), \\
\label{eq:regmodel1_theta}	
\theta &= -c_1.
\end{align}


\subsection{The general case}
Here we construct a regression model for general case~\eqref{eq:y} with $n$ sinusoids.

\begin{prop}
	\label{pr:prop1}
	For any $n \in \mathbb{N}$ the following relation holds
	\begin{align}
		\label{eq:zn}
		&\left[Z^2 + 1 - 2Zc_1\right]\cdot \ldots \cdot \left[Z^2 + 1 - 2Zc_n\right] y^{(n)}(t) = 0,
	\end{align}
	 where $c_i:=\cos \omega_i t$ are constants, $i=\overline{1,n}$, $Z$~is the delay operator~\eqref{eq:delay_operator}, the upper index in the brackets $(n)$ denotes the number of harmonics in the signal~\eqref{eq:y}. 
\end{prop}
\begin{proof}
	We give a proof by induction on $n$. 
	
	\textit{Base case}: The proposition is true for $n = 1$ according to~\eqref{eq:z1}.
	
	\textit{Inductive step}: Show that for any $k > 1$, if the equation~\eqref{eq:zn} for $y^{(k)}(t)$ holds, then for $y^{(k+1)}(t)$ also holds. This can be done as follows.
	
	Assume the induction hypothesis that the equation~\eqref{eq:zn} for $y^{(k)}(t)$ holds. It must then be shown that the equation~\eqref{eq:zn} for $y^{(k+1)}(t)$ holds. Let us express $y^{(k+1)}(t)$ with $y^{(k)}(t)$:
	\begin{align}
	\label{eq:yk}
	y^{(k+1)}(t) =& \sum_{i=1}^{k+1} A_i\sin\left(\omega_i t+\phi_i\right)
	\nonumber\\
	=& \sum_{i=1}^{k} A_i\sin\left(\omega_i t+\phi_i\right) + A_{k+1}\sin\left(\omega_{k+1} t+\phi_{k+1}\right)
	\nonumber\\
	=& y^{(k)}(t) + A_{k+1}\sin\left(\omega_{k+1} t+\phi_{k+1}\right).
	\end{align}
	Using the induction hypothesis, the operator
	\begin{align*}
		Z^{(k)}(\cdot)=\begin{cases}
			0, & t < 2kh,\\
			\left[Z^2 + 1 - 2Zc_1\right]\cdot \ldots \cdot \left[Z^2 + 1 - 2Zc_k\right](\cdot), & t \geq 2kh,
		\end{cases}
	\end{align*}
	can be applied to the equation~\eqref{eq:yk}: 
	\begin{align}
	&\left[Z^2 + 1 - 2Zc_1\right]\cdot \ldots \cdot \left[Z^2 + 1 - 2Zc_k\right]y^{(k+1)}(t) \nonumber\\
	&\quad = \underbrace{\left[Z^2 + 1 - 2Zc_1\right]\cdot \ldots \cdot \left[Z^2 + 1 - 2Zc_k\right]y^{(k)}(t)}_{0}
	\nonumber\\
	&\qquad + \left[Z^2 + 1 - 2Zc_1\right]\cdot \ldots \cdot \left[Z^2 + 1 - 2Zc_k\right]\alpha_{k+1}(t),
	\label{eq:yk1}
	\end{align}
	where $\alpha_{k+1}(t):= A_{k+1}\sin\left(\omega_{k+1} t+\phi_{k+1}\right)$, $t \geq 2kh$.
	
	Note, the function $\alpha_{k+1}(t)$ is the signal $y^{(1)}(t)$ with parameter indexes equal $k+1$. From the base case we obtain
	\begin{align}
		\label{eq:zk1}
		\left[Z^2 + 1 - 2Zc_{k+1}\right] \alpha_{k+1}(t) &= 0, \quad t \geq 2h.
	\end{align}
	
	Applying the operator $\left[Z^2 + 1 - 2Zc_{k+1}\right](\cdot)$ to \eqref{eq:yk1} yields:
	\begin{align}
	&\left[Z^2 + 1 - 2Zc_1\right]\cdot \ldots \cdot \left[Z^2 + 1 - 2Zc_{k+1}\right]y^{(k+1)}(t) \nonumber\\
	&\quad = \left[Z^2 + 1 - 2Zc_1\right]\cdot \ldots \cdot \underbrace{\left[Z^2 + 1 - 2Zc_{k+1}\right]\alpha_{k+1}(t)}_0,
	\nonumber\\
	& \left[Z^2 + 1 - 2Zc_1\right]\cdot \ldots \cdot \left[Z^2 + 1 - 2Zc_{k+1}\right]y^{(k+1)}(t) = 0,
	\end{align}
	which shows that the proposition is true for $y^{(k+1)}(t)$.
	
	Since both the base case and the inductive step have been performed, by mathematical induction the statement holds for all natural numbers $n$. This completes the proof.
\end{proof}

We are now in a position to construct from~\eqref{eq:zn} the regression model for the general case as
\begin{align}
\label{eq:regmodeln}
\psi(t) &= \varphi^\top(t) \theta,
\end{align}
where
\begin{align}
\theta^\top &= \begin{bmatrix}
\theta_1 & \theta_2 & \ldots & \theta_n
\end{bmatrix},\\
\varphi^\top(t) &= \begin{bmatrix}
\varphi_1(t) & \varphi_2(t) & \ldots & \varphi_n(t)
\end{bmatrix},
\end{align}
or more specifically
\begin{align}
	\label{eq:zn2}
	\left[Z^2 + 1\right]^ny(t) = \theta_1 \varphi_1(t)+\theta_2 \varphi_2(t)+\ldots+\theta_n \varphi_n(t).
\end{align}

The regressand $\psi(t)$ is found using Newton's binomial:
\begin{align}
\label{eq:psi_n}
\psi(t) = \left[Z^2 + 1\right]^n y(t)&=\sum_{i=0}^n C_n^i Z^{2(n-i)}y(t) \nonumber\\ 
&=\sum_{i=0}^n C_n^i y(t-2h(n-i)),
\end{align}
where $C_n^i=\frac{n!}{i!(n-i)!}$.

Components of the vector of unknown parameters $\theta$ are related to $c_{i}$ via Vieta's formulas:
\begin{align}
\label{eq:theta_n}
\theta_1 &= -c_1 - c_2 -\ldots - c_n, \\
\theta_2 &= c_1 c_2 + c_1 c_3 +\ldots + c_{n-1} c_n, \\
&\cdots \nonumber\\
\theta_n &= (-1)^n c_1 c_2 \cdot \ldots \cdot c_n.
\end{align}

Components of the regressor $\varphi(t)$ are the following:
\begin{align}
	\label{eq:phi_n}
	\varphi_1(t) &= 2Z\left[Z^2+1\right]^{n-1}y(t)
	= 2Z\sum_{i=0}^{n-1} C_{n-1}^i Z^{2(n-i)}y(t)
	\nonumber \\
	&= 2\sum_{i=0}^{n-1} C_{n-1}^i y(t-h(2(n-i)-1)), \\
	\varphi_2(t) &= 2^2 Z^2\left[Z^2+1\right]^{n-2}y(t)=2^2 Z^2\sum_{i=0}^{n-2} C_{n-2}^i Z^{2(n-i)}y(t)
	\nonumber \\
	&= 2^2\sum_{i=0}^{n-2} C_{n-2}^i y(t-2h(n-i-1)), \\
	&\cdots \nonumber\\
	\varphi_n(t) &= 2^n Z^ny(t) = 2^n y(t-nh). \label{eq:phi_nn}
\end{align}

\begin{prop}
	\label{pr:phi}
	The regressor components $\varphi_1(t)$, $\dots$, $\varphi_n(t)$ from~\eqref{eq:phi_n}--\eqref{eq:phi_nn} are linearly independent functions at $t \geq (2n-1)h$.
\end{prop}
\begin{proof}
	The signal~\eqref{eq:y} consist $n$ harmonics and can be represented as a linear combination of $2n$ signals $y(t-mh )$, where $m = \overline{1,2n}$, how it follows from the proposition~\ref{pr:prop1}.
	
	We use  proof by contradiction to show that signals $y(t-mh )$ at $m = \overline{1,2n}$ are linearly independent. Suppose to the contrary, then the coefficients $a_m \in \mathbb{R}$, $m = \overline{1,2n}$ are not equal to zero at one time, i.~e. $a_1^2+\cdots+a_{2n}^2 \neq 0$ such that the linear combination of signals $y(t-h )$, $\dots$, $y(t-2nh)$ is equal to zero at~$\forall t \geq 2nh$:
	
	\begin{align}
		\nonumber	
		&\sum_{m = 1}^{2n} a_m y(t-mh ) = 0, \\
		\nonumber
		&\sum_{m = 1}^{2n} a_m \sum_{i=1}^{n} A_i\sin\big(\omega_i (t-mh ) + \phi_i\big) = 0, \\
		\nonumber
		&\sum_{m = 1}^{2n} a_m \sum_{i=1}^{n} A_i\sin\big(\omega_i t + \phi_i\big)\cos(mh \omega_i) -\\
		\label{eq:sum}
		& \quad \quad -\sum_{m = 1}^{2n} a_m \sum_{i=1}^{n} A_i\cos\big(\omega_i t + \phi_i\big)\sin(mh \omega_i) = 0.
	\end{align}

	The equation~\eqref{eq:sum} must be true for any values of $A_i$, $\omega_i$ and $\phi_i$, therefore the following equalities for~$i=\overline{1,n}$ must be satisfied:
	\begin{align}	
	\label{eq:sys1}
	\sum_{m = 1}^{2n} a_m \cos(mh \omega_i) =& 0,\\
	\label{eq:sys2}
	\sum_{m = 1}^{2n} a_m \sin(mh \omega_i) =& 0,
	\end{align}	 
	or in a matrix form:
	\begin{align}
	\label{eq:matrix}	
	\underbrace{\begin{bmatrix}
		\cos(h \omega_1) & \ldots & \cos(2nh \omega_1) \\
		\vdots & \ddots & \vdots \\
		\cos(h \omega_n) & \ldots & \cos(2nh \omega_n) \\
		\sin(h \omega_1) & \ldots & \sin(2nh \omega_1) \\
		\vdots & \ddots & \vdots \\
		\sin(h \omega_n) & \ldots & \sin(2nh \omega_n) 
		\end{bmatrix} }_A \cdot
	\begin{bmatrix}
	a_1 \\
	\vdots\\
	a_{2n}
	\end{bmatrix} = 
	\begin{bmatrix}
	0 \\
	\vdots\\
	0
	\end{bmatrix}.
	\end{align}	
	
	The columns of the matrix $A$ are linearly independent and $\text{rank}A=2n$, from constraint~\eqref{eq:h}. Consequently, the homogeneous system of linear equations~\eqref{eq:matrix} has a single zero solution $a_m = 0$, at $m=\overline{1,2n}$, which contradicts the assumption that signals $y(t-mh )$ are linearly dependent at $m = \overline{1,2n}$ . 
	
	Note that the regressor components $\varphi_i(t)$, $i=\overline{1,n}$ are represented as linear combinations of signals $y(t-mh )$, $m = \overline{1,2n-1}$: 
	\begin{align*}
	\begin{bmatrix}
	\varphi_1(t) \\
	\varphi_2(t) \\
	\varphi_3(t) \\
	\vdots \\
	\varphi_n(t)
	\end{bmatrix} &= B
	\cdot
	\begin{bmatrix}
	y(t-h ) \\
	y(t-2h ) \\
	y(t-3h ) \\
	\vdots \\
	y(t-(2n-1)h ) 
	\end{bmatrix},
	\end{align*}	
	where $B$~---~trapezoidal matrix $n \times 2n-1$, $\text{rank}B = n$:
	\begin{align*}
	B = \begin{bmatrix}
	1 & 0 & C_{n-1}^{1}	& 0 & C_{n-1}^{2} & \ldots & 0 & \ldots\\
	0 & 1 & 0 & C_{n-2}^{1}	& 0 & \ldots & 1 & \ldots\\
	0 & 0 & 1 & 0 & C_{n-3}^{1}	& \ldots & 0 & \ldots\\
	\vdots & \vdots & \vdots & \vdots & \vdots & \ldots & \vdots & \ldots \\
	0 & 0 & 0 & 0 & 0 & \ldots & 1 & \ldots\\
	\end{bmatrix}.
	\end{align*}
	
	Consequently, the regressor components are linearly independent functions at $t \geq (2n-1)h$.
\end{proof}


\section{Regression model decomposition}
\label{sec:DREM}

In the previous section, regression model \eqref{eq:regmodeln} is constructed. In this section DREM method~\cite{Aranovskiy2016} is used to obtain $n$ separate first order linear regression models.

Following DREM procedure we introduce new delay operator similarly to~\eqref{eq:delay_operator}:
\begin{align}
\label{eq:filter_delay}
[H_d(\cdot)](t) = \begin{cases}
0, & t < d,\\
(\cdot)(t-d), & t \geq d,
\end{cases}
\end{align}
where $d \in \mathbb{R}_{+}$ is the delay.

Let's apply delay operators to linear regression model~\eqref{eq:regmodeln}:
\begin{align}
\label{eq:regmod_delay}
H^i\{\psi(t)\} = H^i\{\varphi(t)\}^T \theta,\quad i = \overline{1,n},
\end{align}
where $H^i\{ \cdot \} = \underbrace{H\{ H \{ \dots \{ H}_i \{\cdot\} \} \dots \} \}$.

Next, let us introduce a few variables: $\psi_{i}(t):=H^i\{\psi(t)\}$, $\Phi_{i}(t):=H^i\{\varphi(t)\}$ and write the extended system from the expressions \eqref{eq:regmodeln} in a matrix form:
\begin{align}
\label{eq:regmod_ext}
\varepsilon \Psi_f(t) =& \varepsilon \Phi_f(t) \theta, \\
\Psi_f(t) = \begin{bmatrix}
\psi_1(t)\\
\psi_2(t)\\
\ldots \\
\psi_n(t)\\
\end{bmatrix} \in \mathbb{R}^{n},
\quad
\label{eq:phi_extended}
&\Phi_f(t) = \begin{bmatrix}
\Phi_1^T(t)\\
\Phi_2^T(t)\\
\ldots \\
\Phi_n^T(t)
\end{bmatrix} \in \mathbb{R}^{n \times n},
\end{align}
where $\varepsilon~\in~\mathbb{R}_+$ is normalization gain.

At the mixing step of DREM procedure the regression model~\eqref{eq:regmod_ext} multiply to the adjugate matrix $\mbox{adj}\{\Phi_f(t)\}$ and we get
\begin{align}
\label{eq:regmodel_drem}
\Psi(t) =& \Delta(t) \theta, \\
\Psi(t) :=& \mbox{adj}\{\varepsilon \Phi_f(t)\}\varepsilon \Psi_f(t) = \begin{bmatrix}
\Psi_1(t) \\
\Psi_2(t) \\
\ldots \\
\Psi_n(t)
\end{bmatrix}, \\
\label{eq:Delta}
\Delta(t) :=& \mbox{det}\{\varepsilon \Phi_f(t)\},
\end{align}
where $\mbox{adj}\{\cdot\}$ is the adjugate matrix, $\mbox{det}\{\cdot\}$ is the determinant.

Let's rewrite the equation~\eqref{eq:regmodel_drem} in a component way:
\begin{align}
\label{eq:regmod_final}
\Psi_i(t) &=  \Delta(t) \theta_i, \quad i=\overline{1,n}.
\end{align}
where $\Delta(t) \in \mathbb{R}$, $\Psi(t) \in \mathbb{R}$.

Now we can estimate parameters $\theta_i$ from~\eqref{eq:regmod_final} separately.


\section{Finite-time parameter estimation}
\label{sec:estimator}

Parameters estimations of the first order regression model~\eqref{eq:regmod_final} can be obtained using the standard gradient method~\cite{ioannou1996robust}:
\begin{align}
	\label{eq:grad}
	\dot{\hat{\theta}}_i(t) & =\gamma_i\Delta(t)\left(\Psi_i(t)-\Delta(t)\hat{\theta}_i(t)\right),
\end{align}
where $\hat{\theta}_i\in\mathbb{R}$ is estimate of $\theta_i$, $\gamma_i\in\mathbb{R}_{+}$
is a tuning gain.

Now we can apply the modification for the scalar standard gradient method, which is proposed in~\cite{ortega2019fto}.

The error model for $\tilde{\theta}_i(t) =\theta_i-\hat{\theta}_i(t)$ is expressed as
\begin{align}
	\label{eq:error_model}
	\dot{\tilde{\theta}}_i(t) & =-\gamma_i \Delta^{2}(t)\tilde{\theta}_i(t).
\end{align}

One can easily find the solution for~\eqref{eq:error_model}:
\begin{align}
	\label{eq:error_solution}
	\tilde{\theta}_i(t) &=\tilde{\theta}_i(0)\text{e}^{-\gamma_i\int_{0}^{t}\Delta^{2}(r )\text{d} r}.
\end{align}

\begin{prop}
	\label{pr:exp}
	The algorithm~\eqref{eq:grad} provides an exponential convergence of the estimation error~$\tilde{\theta}_i(t)$ to zero at~$i = \overline{1,n}$.
\end{prop}
\begin{proof}
	If the function $\Delta(t)$ is bounded and persistently exciting, then the estimation method~\eqref{eq:grad} provides an exponential convergence of the estimation error $\tilde{\theta}_i(t)$ to zero according to~\cite{ioannou1996robust}. 
	
	The function $\Delta(t)$ is the determinant of the extended regressor matrix~\eqref{eq:phi_extended} and each it element is represented as a linear combination of continuous and limited harmonic functions~\eqref{eq:phi_n}--\eqref{eq:phi_nn}. Consequently, $\Delta(t)$ is a bounded function.	
	
	Let's show that the $\Delta(t)$ is persistently exciting, i.e. there exist the positive constants $T$, $t_0$, $a \in \mathbb{R}_+$ such that 
	\begin{align}
	\label{eq:pe}
	\int_{t}^{t+T} \Delta^2(r) \text{d} r \geq a,\quad \forall t > t_0.
	\end{align}
	
	Now let's write the function $\Delta(t)$ using equations~\eqref{eq:Delta}, \eqref{eq:phi_extended} è \eqref{eq:phi_n}--\eqref{eq:phi_nn}:
	\begin{align}
	\nonumber
	\Delta(t) :=& \mbox{det}\{\varepsilon \Phi_f(t)\} = 
	\varepsilon^n \cdot \begin{vmatrix}
	\Phi_1^T(t)\\
	\Phi_2^T(t)\\
	\ldots \\
	\Phi_n^T(t)
	\end{vmatrix} = \\
	\label{eq:Delta_re}
	=& \varepsilon^n \cdot \underbrace{\begin{vmatrix}
		H\{\varphi_1(t)\} &  H\{\varphi_2(t)\} & \cdots & H\{\varphi_n(t)\} \\
		H^2\{\varphi_1(t)\} &  H^2\{\varphi_2(t)\} & \cdots & H^2\{\varphi_n(t)\} \\
		\vdots & \vdots & \ddots & \vdots \\
		H^n\{\varphi_1(t)\} &  H^n\{\varphi_2(t)\} & \cdots & H^n\{\varphi_n(t)\}
		\end{vmatrix}}_C.
	\end{align}
	
	The functions $\varphi_1(t), \ \dots, \ \varphi_n(t)$ are linear independent, according the proposition~\ref{pr:phi}. Then the functions $H^i\{\varphi_1(t), \ \dots, H^i\{\varphi_n(t)$ are also linearly independent due to the replacement of $t_i:= t - ih $. Thus the columns of the matrix $C$ are linearly independent. 
	
	Note that the components of the matrix $C$ are periodic functions, then the determinant $\Delta(t)$ is also a periodic function with the period $T_0$.
	
	Due to the continuity of $\Delta(t)$ and $\text{rang}C = n$ for $\forall t > n(h + d)$, for some $t_1>n(h + d)$ there is $\hat{a} \in \mathbb{R}_+$ that following inequality is true: 
	\begin{align}
	\label{eq:pe_proof}
	\int_{t_1}^{t_1+T_0} \Delta^2(r) \text{d} r \geq \hat{a}.
	\end{align}
	
	Then, the condition~\eqref{eq:pe} is satisfied for $\forall t > n(h + d)$ and $a:=\hat{a}$, $T:=T_0$ due to the periodicity of the subintegral function.
\end{proof}

Replacing $\tilde{\theta}_i(t)$ with $\theta_i-\hat{\theta}_i(t)$ gives the following relation
\begin{align}
	\label{eq:error_solution_opened}
	\theta_i - \hat{\theta}_i(t) & =\theta_i W(t)-\hat{\theta}_i(0)W(t),
\end{align}
where $W(t):=\text{e}^{-\gamma_i\int_{0}^{t}\Delta^{2}(r )\text{d} r}$.

Now from \eqref{eq:error_solution_opened} parameter $\theta_i$ can be
explicitly found
\begin{align}
	\theta_i-\theta_i W(t) & =\hat{\theta}_i(t)-\hat{\theta}_i(0)W(t),\\
	\theta_i\left(1-W(t)\right) & =\hat{\theta}_i(t)-\hat{\theta}_i(0)W(t),
\end{align}
and for some $t_{ft} > n(h + d)$:
\begin{align}
	\hat{\theta}_i^{ft}(t) =& \frac{1}{1-W(t)}\left(\hat{\theta}_i(t)-\hat{\theta}_i(0)W(t)\right), \quad t \geq t_{ft}.
	\label{eq:theta_explicit}
\end{align}
where $\hat{\theta}_i^{ft}(t)$ is the finite-time estimation of the parameter $\theta_i$, $i=\overline{1,n}$.

Using parameter estimates~\eqref{eq:theta_explicit}, we can reconstruct $c_i=\cos(\omega_i h)$ using Vieta's formulas \eqref{eq:theta_n}:
\begin{align}
\label{eq:theta_est}
\hat{\theta}_1^{ft}(t) &= -\hat{c}^{ft}_1(t) - \hat{c}^{ft}_2(t) -\ldots - \hat{c}^{ft}_n(t), \\
\hat{\theta}_2^{ft}(t) &= \hat{c}^{ft}_1(t) \hat{c}^{ft}_2(t) + \hat{c}^{ft}_1(t) \hat{c}^{ft}_3(t) +\ldots + \hat{c}^{ft}_{n-1}(t) \hat{c}^{ft}_n(t), \\
&\cdots \nonumber\\
\label{eq:theta_estn}
\hat{\theta}_n^{ft}(t) &= (-1)^n \hat{c}^{ft}_1(t) \hat{c}^{ft}_2(t) \cdot \ldots \cdot \hat{c}^{ft}_n(t).
\end{align}

Finally, using identity $c_i = \cos\omega_i t$, we obtain the desired frequency estimates~\eqref{eq:objective} as
\begin{align}
	\label{eq:omega}
	\hat{\omega}^{ft}_i(t) = \frac{1}{h}\arccos(\hat{c}^{ft}_i(t)), \quad i=\overline{1,n},
\end{align}
at the predefined time $t_{ft}$.

An alternative method of building frequency estimates $\hat{\omega}^{ft}_i(t)$ using $\hat{c}^{ft}_i$ without the inverse trigonometric function $\arccos \{\cdot\}$ is described in the paper~\cite{vedyakov2018frequency}.

The estimation algorithm~\eqref{eq:theta_explicit} guarantees the convergence of frequency estimates  $\hat{\omega}_i^{ft}(t)$ to the true values $\omega_i$, $i=\overline{1,n}$ for the finite time $t_{ft}$.

\begin{remark}
	\label{re:grad}
	From expressions~\eqref{eq:grad}, \eqref{eq:theta_est}--\eqref{eq:theta_estn}, \eqref{eq:omega} can be obtained the estimates $\hat{\omega}^{grad}_i(t)$, which provides exponential convergence of the estimation errors to zero~\cite{vedyakov2018frequency}.
\end{remark}


\section{Simulation Results}
\label{sec:example}

In this section, we present simulation results that illustrate the efficiency of the proposed estimation algorithm. All simulations have been performed in MATLAB Simulink.

We compare four frequency estimation algorithms:
\begin{itemize}
	\item the proposed finite time algorithm and denote the frequency estimates obtained by this method as~$\hat{\omega}^{ft}_i(t)$;
	\item the gradient method with DREM referred to the remark~\ref{re:grad}, denote the frequency estimates as $\hat{\omega}^{grad}_i(t)$; 
	\item an adaptive observer approach~\cite{pin2019identification} with frequency estimates $\hat{\omega}^{obs}_i(t)$; 
	\item a filtering approach~\cite{pyrkin2015Estimation} and denote the frequency estimates obtained by this method as~$\hat{\omega}^{filt}_i(t)$.
\end{itemize}


\subsection{The case of a two harmonics in noiseless scenario}

In this example, the input signal is composed by two sinusoids:
\begin{align*}
y(t) = \sin 2t + \cos 3t.
\end{align*}

All the methods are initialized with the same initial condition $\hat{\omega}(0) = [2 \ 5]^T$ and are tuned to ensure similar convergence speed in the absence of disturbance. More specifically, the parameters of the proposed estimator and the gradient method with DREM are chosen as: $h = 0.1$~s, $d = 0.13$~s, $\varepsilon = 100$, $\gamma_1 = \gamma_2 = 0.005$. The method in~\cite{pin2019identification} is tuned with: $\bar{g} = 4$, $r = 1.2$, $\mu = 10$, $\Lambda = [1 \ 3]$. The tuning parameters of the method in~\cite{pyrkin2015Estimation} are set as $\lambda = 5$, $\omega_0 = 0.5$, $k_1 = 1$, $k_2 = 5$. The behavior of the estimators is shown in Fig.~\ref{fig:omega}. As it can be seen, all the methods succeed in detecting the frequencies with similar transient time.

\begin{figure}
	\begin{center}
		\begin{subfigure}[t]{\linewidth}
			\includegraphics[width=\linewidth]{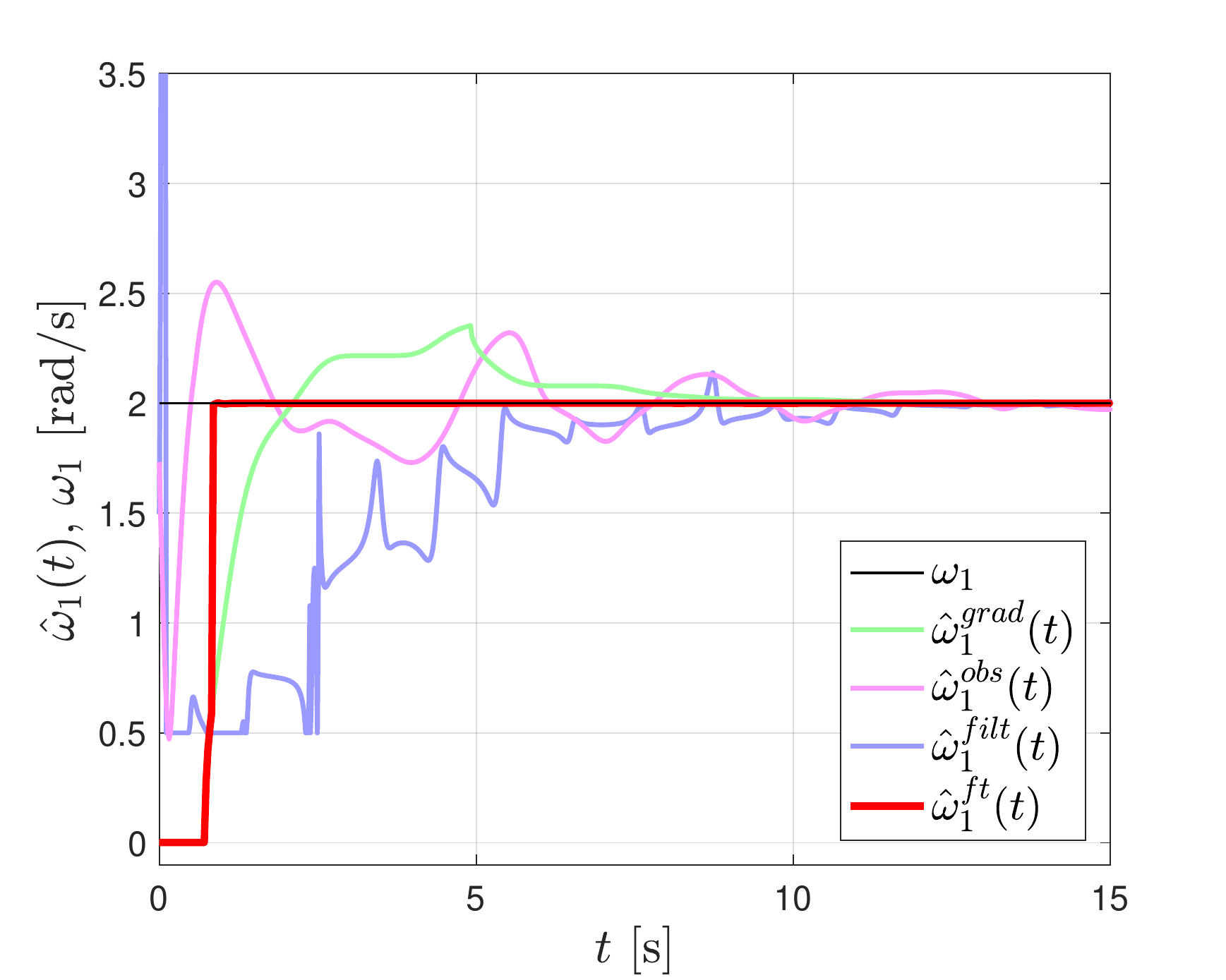}
			\caption{Time behavior of the estimated frequency $\omega_1$}
			\label{fig:om1}
		\end{subfigure}
		
		\begin{subfigure}[t]{\linewidth}
			\includegraphics[width=\linewidth]{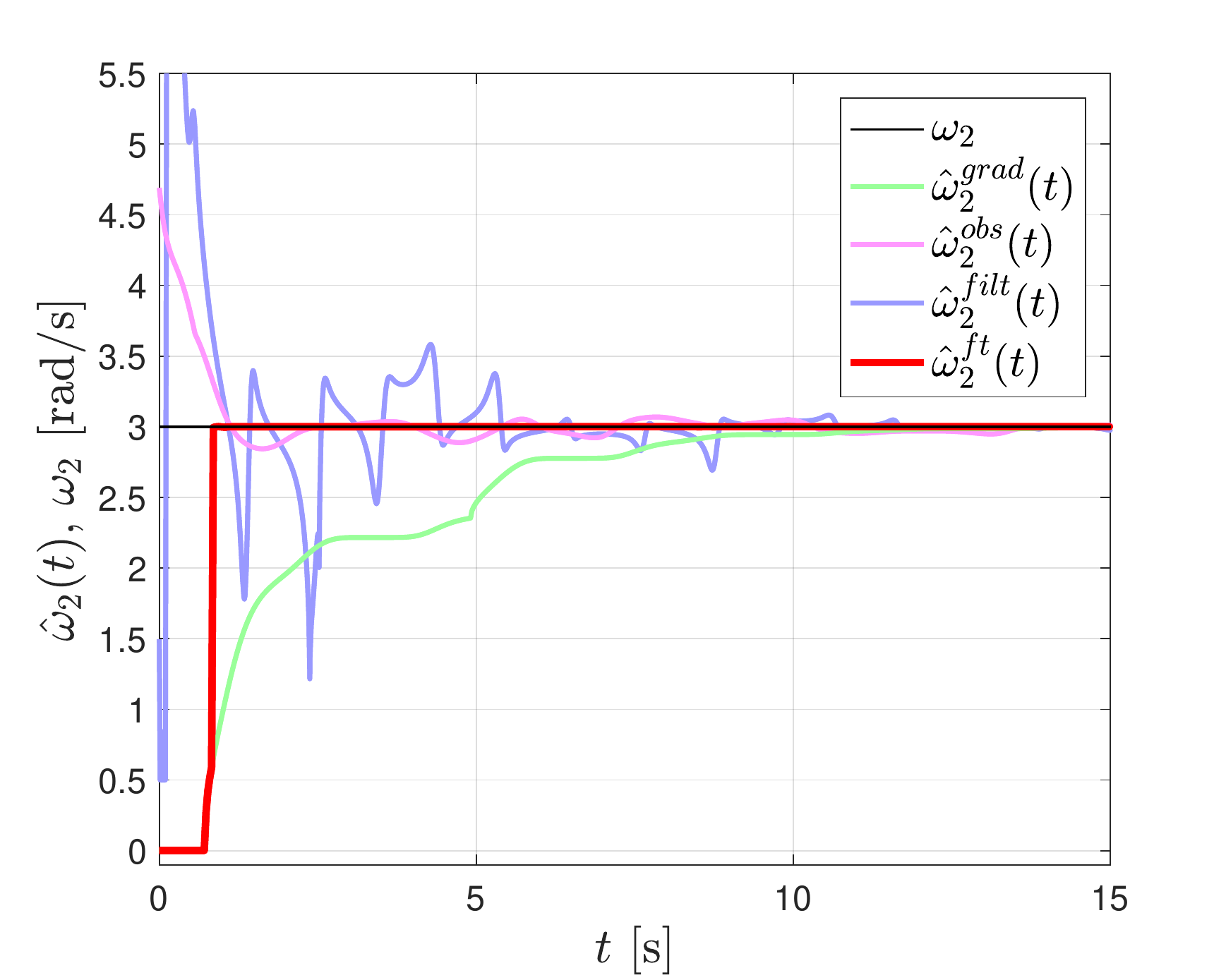}
			\caption{Time behavior of the estimated frequency $\omega_2$}
			\label{fig:om2}
		\end{subfigure}
		
		\caption{Simulation results for the two harmonic case for an unperturbed multi-sinusoidal signal}
		\label{fig:omega}
	\end{center}
\end{figure}

\subsection{The case of a two harmonics in noisy scenario}

In this example, the four algorithms are compared in presence of a  disturbance $d(t)$, which is added to the measurement: $\hat{y}(t) = \sin 2t + \sin 3t + d(t)$.

Let's consider the influence of harmonic disturbance $d(t) = 0.25\sin 15t$ on the obtained frequency estimates. The parameters of the estimation methods are the following:
\begin{itemize}
	\item $\hat{\omega}^{ft}(t)$, $\hat{\omega}^{grad}(t)$: $h = 1$, $d = 0.37$, $\varepsilon = 0.1$, $\gamma_1 = \gamma_2 = 1$;
	\item $\hat{\omega}^{obs}(t)$: $\bar{g} = 3$, $r = 1.2$, $\mu = 2$, $\Lambda = [1 \ 5]$;
	\item $\hat{\omega}^{filt}(t)$: $\lambda = 0.7$, $\omega_0 = 0.5$, $k_1 = 20$, $k_2 = 75$.
\end{itemize}

The behavior of the estimators is shown in Fig.~\ref{fig:omega_sin}. 
\begin{figure}
	\begin{center}
		\begin{subfigure}[t]{\linewidth}
			\includegraphics[width=\linewidth]{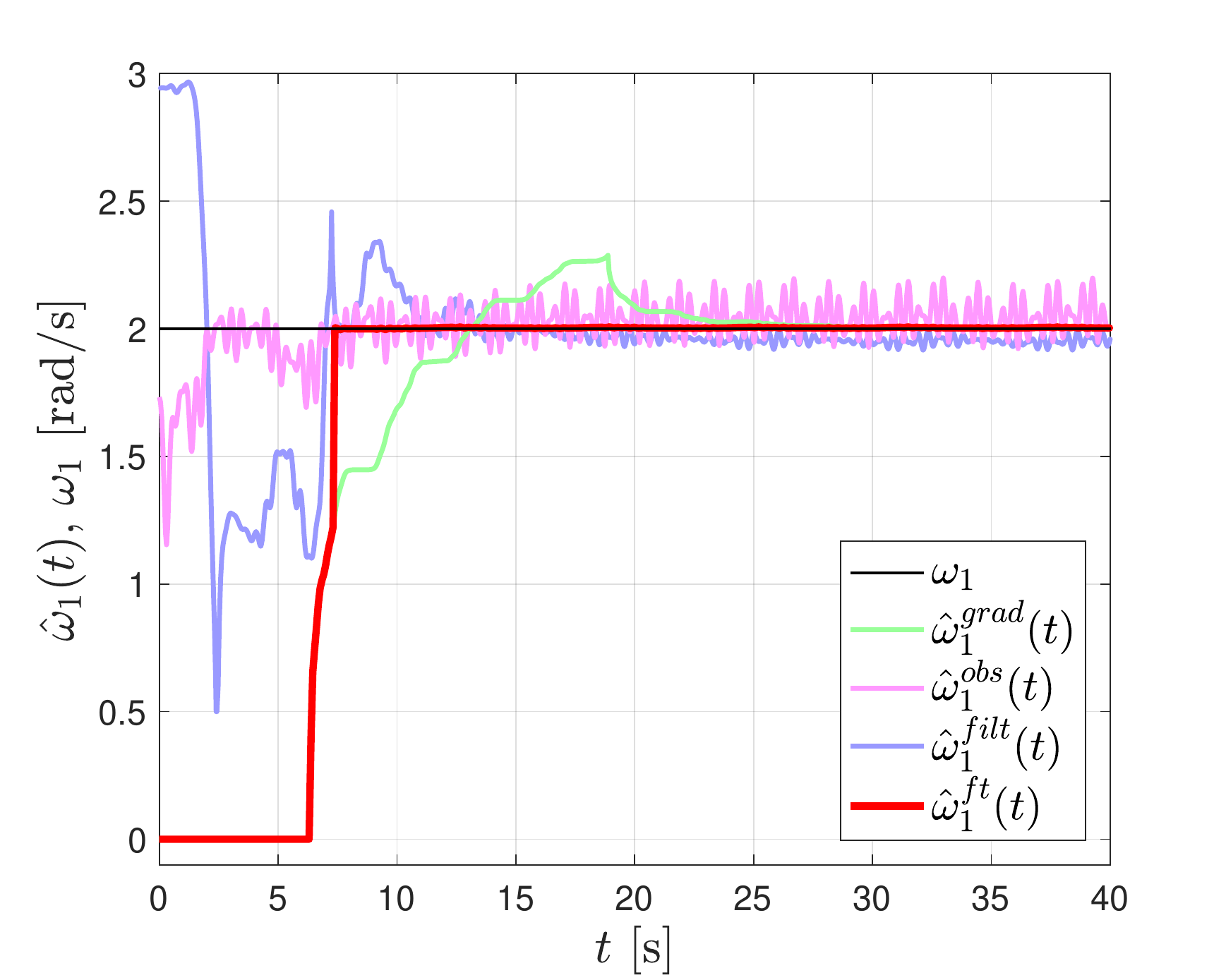}
			\caption{Time behavior of the estimated frequency $\omega_1$}
			\label{fig:om1_sin}
		\end{subfigure}
		
		\begin{subfigure}[t]{\linewidth}
			\includegraphics[width=\linewidth]{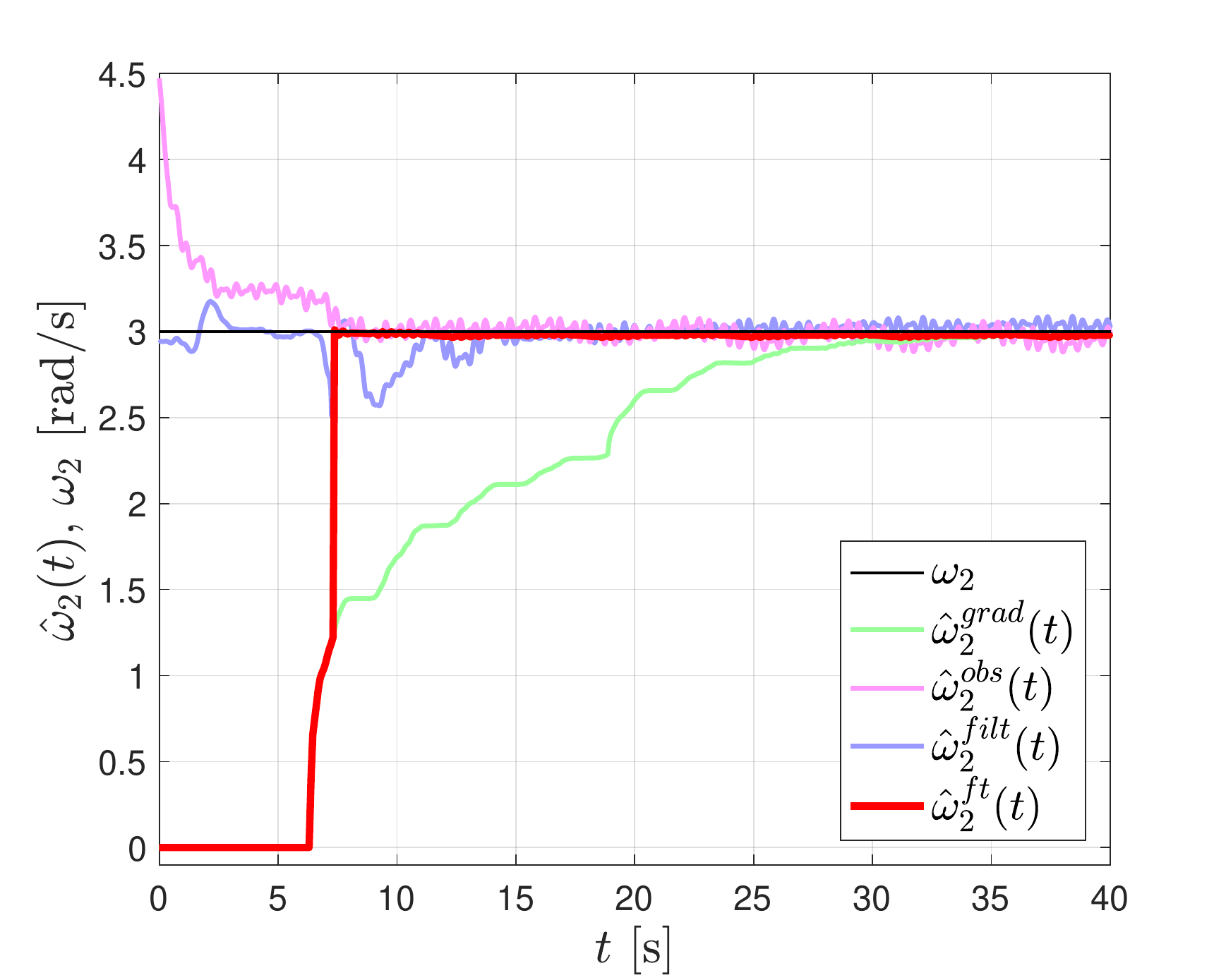}
			\caption{Time behavior of the estimated frequency $\omega_2$}
			\label{fig:om2_sin}
		\end{subfigure}
		
		\caption{Simulation results for the two harmonic case in harmonic noisy scenario}
		\label{fig:omega_sin}
	\end{center}
\end{figure}

As another type of external perturbation, let us consider the random noise. Additive noise $d(t)$ is simulated as a uniformly distributed process ranging within $[-0.2, 0.2]$ and sample time $0.001$ s. The parameters of the estimation methods in this case are the following:
\begin{itemize}
	\item $\hat{\omega}^{ft}(t)$, $\hat{\omega}^{grad}(t)$: $h = 0.6$, $d = 0.4$, $\varepsilon = 0.1$, $\gamma_1 = \gamma_2 = 1$;
	\item $\hat{\omega}^{obs}(t)$: $\bar{g} = 4$, $r = 1.2$, $\mu = 3$, $\Lambda = [1 \ 4]$;
	\item $\hat{\omega}^{filt}(t)$: $\lambda = 2$, $\omega_0 = 1$, $k_1 = 0.2$, $k_2 = 0.5$.
\end{itemize}
The behavior of the estimators is shown in Fig.~\ref{fig:omega_noise}.
\begin{figure}
	\begin{center}
		\begin{subfigure}[t]{\linewidth}
			\includegraphics[width=\linewidth]{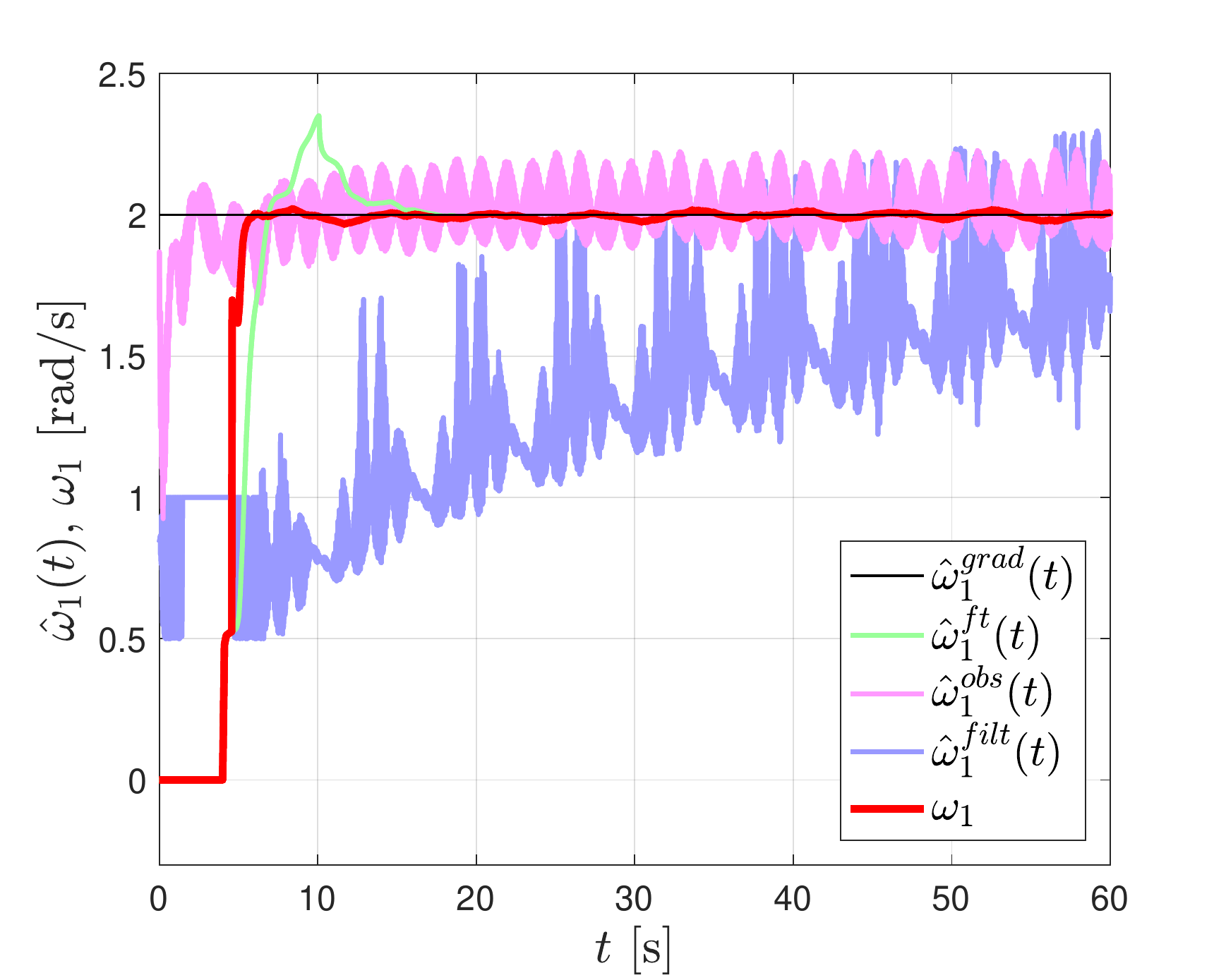}
			\caption{Time behavior of the estimated frequency $\omega_1$}
			\label{fig:om1_noise}
		\end{subfigure}
		
		\begin{subfigure}[t]{\linewidth}
			\includegraphics[width=\linewidth]{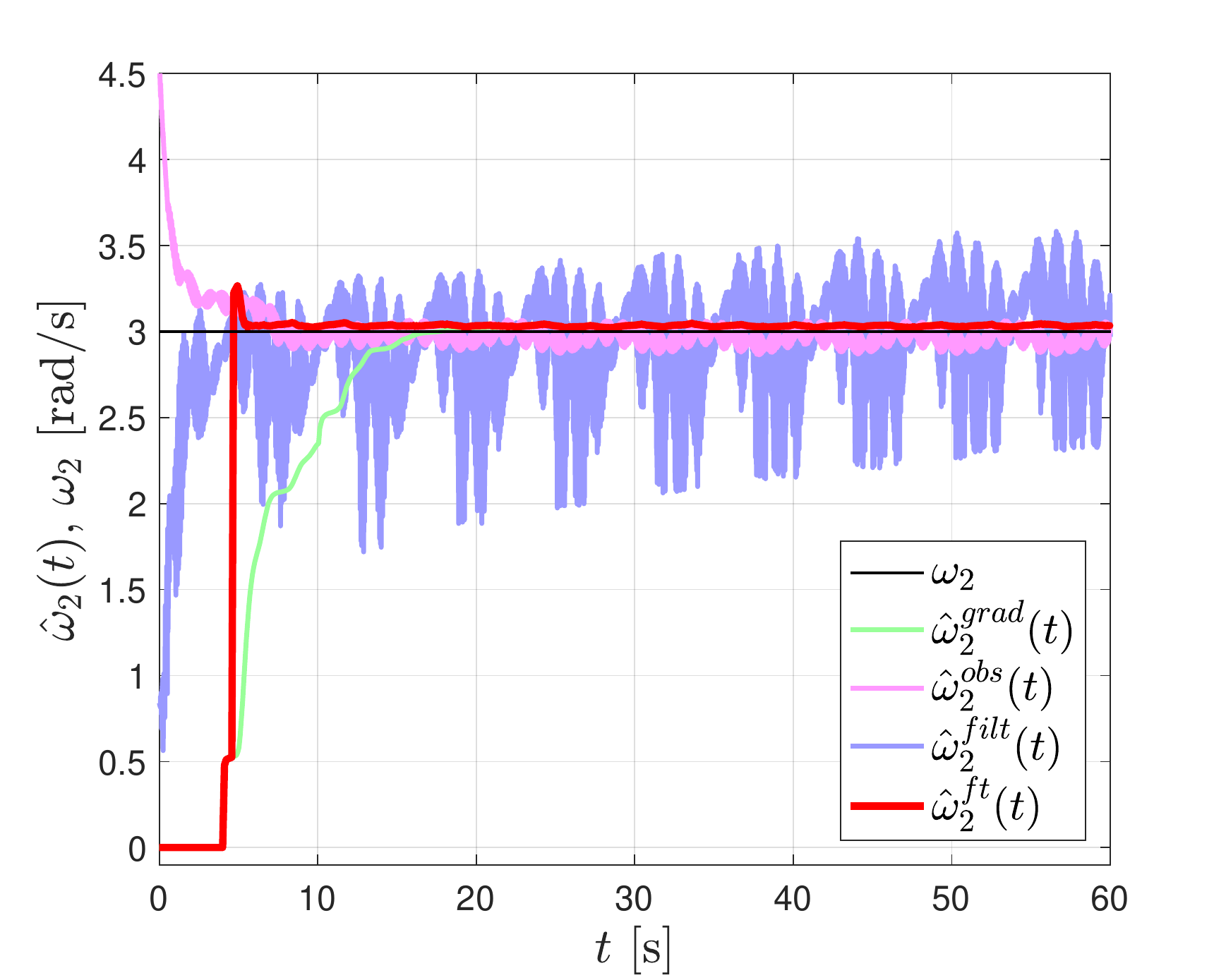}
			\caption{Time behavior of the estimated frequency $\omega_2$}
			\label{fig:om2_noise}
		\end{subfigure}
		
		\caption{Simulation results for the two harmonic case in random noisy scenario}
		\label{fig:omega_noise}
	\end{center}
\end{figure}

As it can be noticed from Fig.~\ref{fig:omega_sin}--\ref{fig:omega_noise}, all the methods are capable to track both the frequencies in noisy scenario. However, the proposed finite time algorithm is more susceptible to high frequency and random disturbances than the other three methods. The proposed algorithm still provides a finite-time estimate in the presence of additive disturbance.


\subsection{The case of step-wise frequency variation}

Let us consider a multi-sinusoidal signal whose components exhibit a step-wise frequency variation:
\begin{align*}
y(t) &= \sin \omega_1(t) + \cos \omega_2(t),\\
\omega_1(t) &= \begin{cases}
1.8, & t < 30 \ s,\\
2, & t \geq 30 \ s,
\end{cases} \quad 
\omega_2(t) = \begin{cases}
3.2, & t < 30 \ s,\\
3, & t \geq 30 \ s.
\end{cases}
\end{align*}

The parameters are the following:
\begin{itemize}
	\item $\hat{\omega}^{ft}(t)$, $\hat{\omega}^{grad}(t)$: $h = 0.7$, $d = 0.4$, $\varepsilon = 0.1$, $\gamma_1 = \gamma_2 = 1$;
	\item $\hat{\omega}^{obs}(t)$: $\bar{g} = 4$, $r = 1.2$, $\mu = 10$, $\Lambda = [1 \ 4]$;
	\item $\hat{\omega}^{filt}(t)$: $\lambda = 2$, $\omega_0 = 0.5$, $k_1 = 1$, $k_2 = 5$.
\end{itemize}

The results are reported in Fig.~\ref{fig:omega_swit}. As it can be noticed, all four methods favorably deal with the frequency change and lead to comparable stationary behavior. However, after the change of frequencies of the measured signal the proposed method  can not produce frequencies estimates at the finite time. The estimation plots of $\hat{\omega}_i^{ft}(t)$ and $\hat{\omega}_i^{grad}(t)$ match and there is an exponential convergence.This aspect results from using  $\hat{\theta}_i(0)$ in the expression \eqref{eq:theta_explicit}.

\begin{figure}
	\begin{center}
		\begin{subfigure}[t]{\linewidth}
			\includegraphics[width=\linewidth]{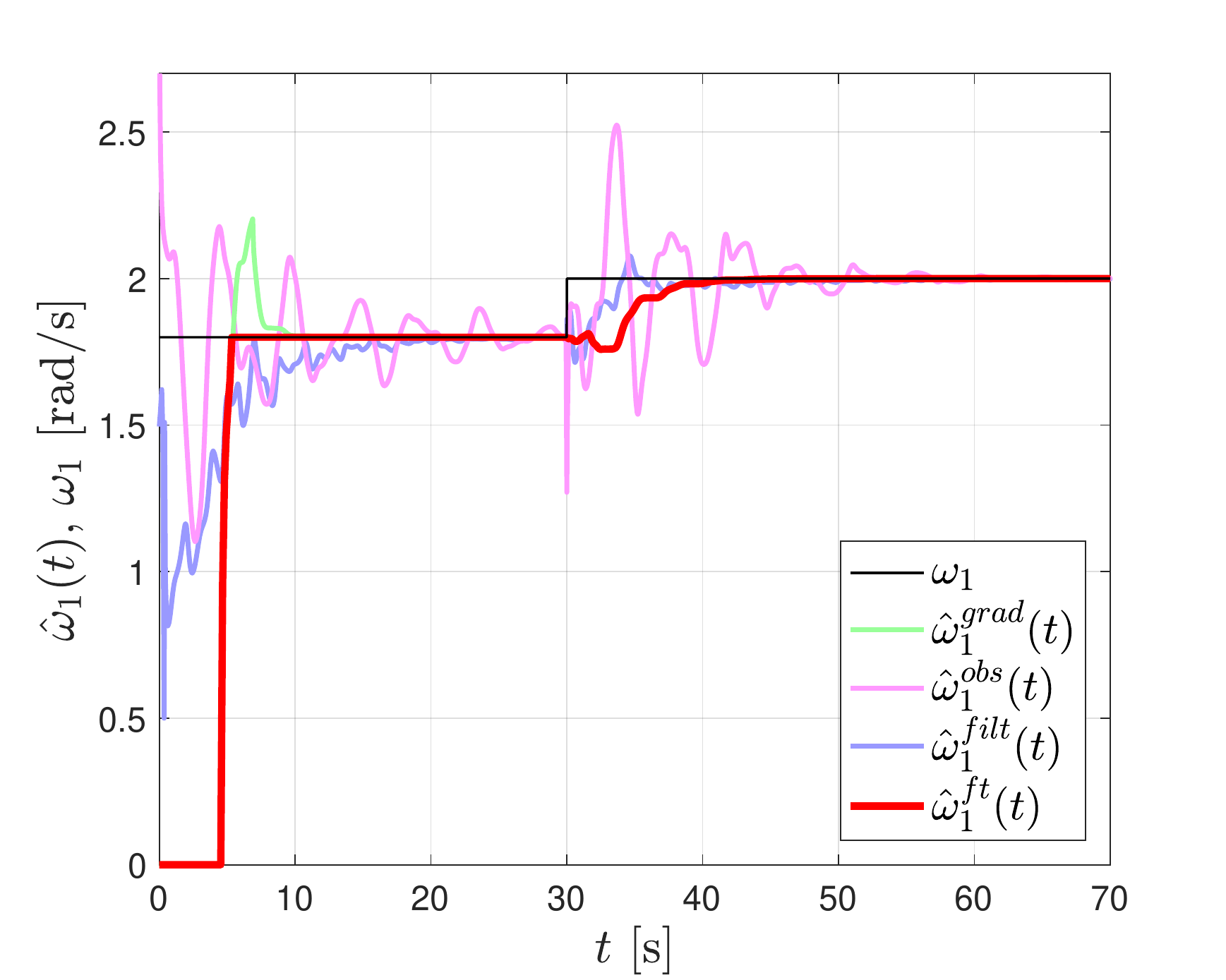}
			\caption{Time behavior of the estimated frequency $\omega_1$}
			\label{fig:om1_switch}
		\end{subfigure}
		
		\begin{subfigure}[t]{\linewidth}
			\includegraphics[width=\linewidth]{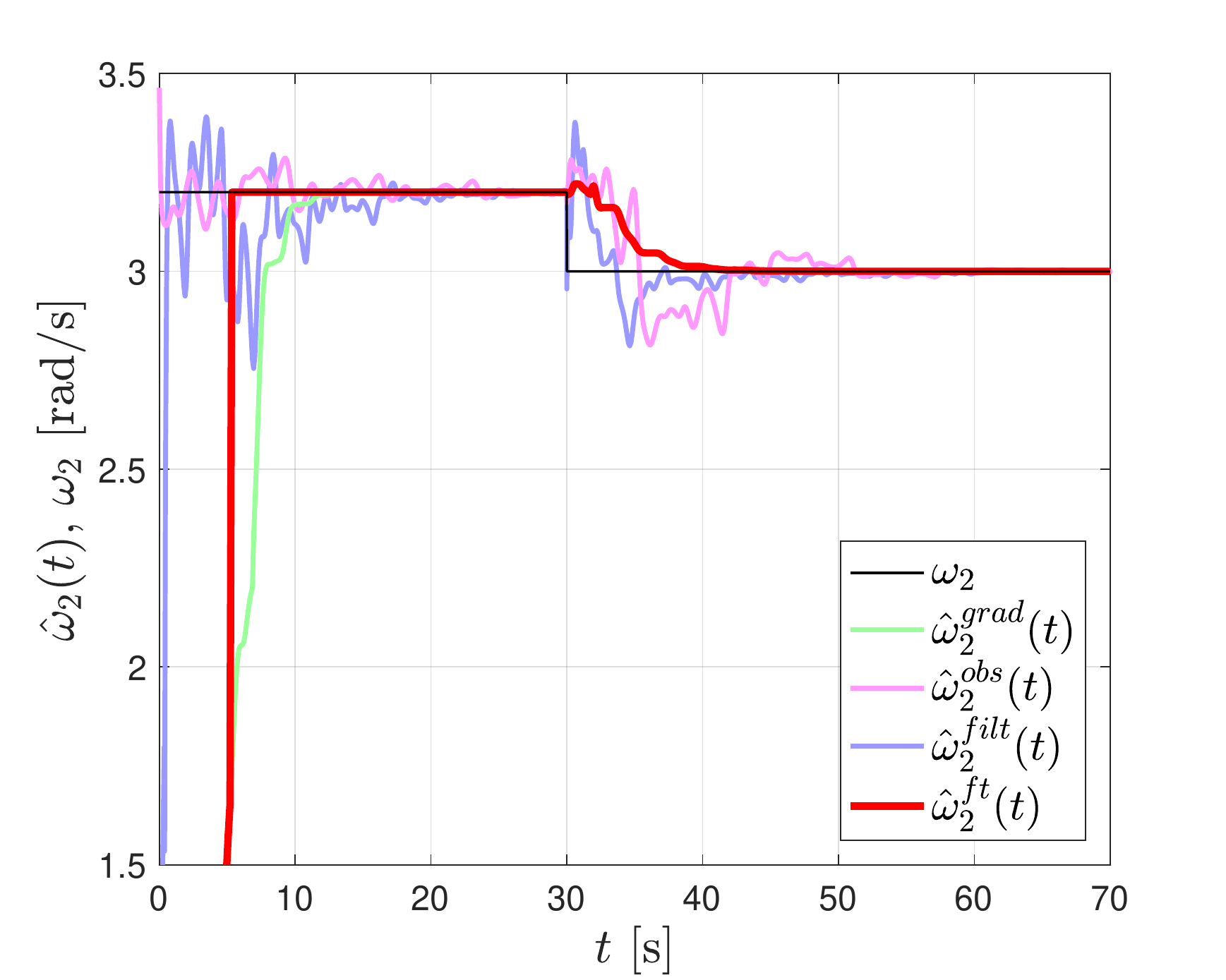}
			\caption{Time behavior of the estimated frequency $\omega_2$}
			\label{fig:om2_switch}
		\end{subfigure}
		
		\caption{Simulation results for an unperturbed two harmonic case with a step frequency change at 30 s}
		\label{fig:omega_swit}
	\end{center}
\end{figure}

\section{Conclusions}
\label{sec:conclusions}

The finite-time frequency estimation problem for an unbiased multi-sinusoidal signal is considered. The presented result extends the previously proposed estimator for a pure sinusoidal signal~\cite{med2017sin}. The $n$-th order linear regression model, where unknown parameters depend on signal frequencies, is constructed using transport delay operators.

Applying DREM, we split the model into scalar regressions. Unknown parameters are estimated by standard gradient method. Using obtained estimates, which converge exponentially to the true values, we calculate at a predefined finite time another estimate using scheme proposed in~\cite{ortega2019fto}.

The relation between estimated parameters and frequencies is quite complicated. The estimated parameters are $n$-th order polynomial coefficients. Using these values, we find polynomial roots, which requires numerical solver for $n \geq 5$. Using $\arccos$ function and dividing by the delay, we reconstruct the frequencies. 

The method provides estimates at a predefined time, so it is independent of frequency values and the tuning gain of the internally used gradient descent method. In the noise scenario, it allows increasing estimation quality without trade-off between noise sensitivity and estimation duration. Decreasing tuning gain in the gradient descent method, we decrease sensitivity to the measurement noise. Without the finite-time estimation scheme, it dramatically increases estimation duration. 


\bibliographystyle{IEEEtran}
\bibliography{IEEEabrv,references}

\end{document}